\documentclass[a4paper,11pt]{article}
\usepackage[]{amsmath}
\usepackage{amssymb}
\usepackage{amsthm}
\usepackage{amsfonts,amssymb}
\usepackage{mathrsfs}
\usepackage{tikz}
\usetikzlibrary{shapes}
\usepackage[T1]{fontenc}
\usepackage{latexsym}
\usepackage{enumitem}
\usepackage{authblk}
\usepackage[sort&compress,numbers]{natbib}
\usepackage{graphicx}
\usepackage{stmaryrd}
\usepackage[]{paralist}
\usepackage{hyperref}
\usepackage[ruled,vlined,linesnumbered]{algorithm2e}
\usepackage[nameinlink,capitalize]{cleveref}
\usepackage{bookmark}
\usepackage{doi}
\usepackage{mathtools}
\usepackage[vmargin=2.85cm,hmargin=3.25cm]{geometry}

\DeclarePairedDelimiter{\ceil}{\lceil}{\rceil}
\DeclarePairedDelimiter{\floor}{\lfloor}{\rfloor}

\newcommand{\mmod}[1]{\ \mathrm{mod}\ #1}

\def\papertitle{Proportionally dense subgraph of maximum size: complexity and approximation}
\title{\papertitle}
\hypersetup{
  bookmarksopen=false,
  pdftitle={\papertitle},
  pdfauthor={},
  pdftoolbar=false, %
  pdfmenubar=true, %
  pdfcreator=\LaTeX,
  pdfpagelayout=SinglePage, %
  pdffitwindow=true, %
  pdfkeywords={},
  colorlinks=true,
  linkcolor={black},
  citecolor={blue!50!black},
  urlcolor={black}
}

\newcommand{\probdef}[6][Question]{
\hbox{\vbox{
\begin{quote}
  \label{#6}
  \ifthenelse{\equal{#4}{}}{}{{#4}\ifthenelse{\equal{#5}{}}{}{ ({#5})}}
  \begin{compactdesc}
    \item [Input] {#2}
    \item [#1] {#3}
  \end{compactdesc}
\end{quote}
}}

}

\crefname{figure}{Figure}{Figures}
\crefname{appendix}{Appendix}{Appendices}

\def\O{\mathcal{O}}
\def\NP{\textsf{NP}}
\def\FPT{\textsf{FPT}}

\def\APX{\textsf{APX}}

\def\MIS{\textsc{Max Independent Set}}
\def\MC{\textsc{Max PDS}}
\def\MPDS{\textsc{Max PDS}}
\def\MClong{\textsc{Max Proportionally Dense Subgraph}}
\def\MPDSlong{\textsc{Max Proportionally Dense Subgraph}}

\newcommand{\overbar}[1]{\mkern 3.5mu\overline{\mkern-3.5mu#1\mkern-1.5mu}\mkern 1.5mu}

\usetikzlibrary{fit,automata,positioning,calc}
\usetikzlibrary{shapes.arrows,patterns,decorations.pathreplacing}

\usepackage{lineno,hyperref}

\newtheorem{theorem}{Theorem}
\newtheorem{lemma}{Lemma}
\newtheorem{corollary}{Corollary}
\newtheorem{definition}{Definition}
\newtheorem{proposition}{Proposition}

\usepackage{stmaryrd}
\SetSymbolFont{stmry}{bold}{U}{stmry}{m}{n}

\begin{document}

\author[1]{Cristina Bazgan
\thanks{\href{mailto:bazgan@lamsade.dauphine.fr}{bazgan@lamsade.dauphine.fr}}}
\author[2]{Janka Chleb\'\i kov\'a
\thanks{\href{mailto:janka.chlebikova@port.ac.uk}{janka.chlebikova@port.ac.uk}}}

\author[2]{Cl\'ement Dallard
\thanks{\href{mailto:clement.dallard@port.ac.uk}{clement.dallard@port.ac.uk}}}

\author[1]{Thomas Pontoizeau
\thanks{\href{mailto:thomas.pontoizeau@lamsade.dauphine.fr}{thomas.pontoizeau@lamsade.dauphine.fr}}}

\affil[1]{Universit\'e Paris-Dauphine, Universit\'e PSL, CNRS,  LAMSADE, 75016 Paris, France}
\affil[2]{School of Computing, University of Portsmouth, Portsmouth, United Kingdom}

\date{}

\maketitle

\begin{abstract}
We define a \emph{proportionally dense subgraph} (PDS) as an induced subgraph of a graph with the property that each vertex in the PDS is adjacent to proportionally as many vertices in the subgraph as in the graph.
We prove that the problem of finding a PDS of maximum size is \APX-hard on split graphs, and \NP-hard on bipartite graphs.
We also show that deciding if a PDS is inclusion-wise maximal is \textsf{co-NP}-complete on bipartite graphs.
Nevertheless, we present a simple polynomial-time $(2-\frac{2}{\Delta+1})$-approximation algorithm for the problem, where $\Delta$ is the maximum degree of the graph.
Finally, we show that all Hamiltonian cubic graphs with $n$ vertices (except two) have a PDS of size $\floor{ \frac{2n+1}{3} }$, which we prove to be an upper bound on the size of a PDS in cubic graphs.
\end{abstract}


\section{Introduction}

For a graph $G=(V,E)$, the density of a subgraph on a vertex set $S \subseteq V$ is commonly defined as $\frac{|E(S)|}{|S|}$, where $E(S)$ is the set of edges in the subgraph.
The problem of finding a subgraph of maximum density can be solved in polynomial time using a max flow technique \cite{Goldberg:CSD-84-171}.
However, when the subgraph must contain \emph{exactly} $k$ vertices, the problem becomes \NP-hard \cite{feige2001dense,asahiro2002complexity} and is known as the \textsc{Densest $k$-subgraph} problem.
Two variants of the problem have also been studied where the number of vertices in the subgraph must be either \emph{at least} $k$ or \emph{at most} $k$.
The former is known to be \NP-hard \cite{KhullerSaha2009}, but there exists a polynomial-time $2$-approximation algorithm to solve it \cite{Anderson2007}.
It was showed that any $\alpha$-approximation for the \emph{at most} $k$ variant would imply a $\Theta (\alpha^{2})$-approximation for the densest $k$-subgraph problem \cite{andersen2009finding}.

An induced subgraph on a vertex set $S \subset V$ is said to be \emph{proportionally dense} if all of its vertices in $S$ have \emph{proportionally} as many neighbors in the subgraph as in the graph, and hence the condition $\frac{d_S(u)}{|S|-1}\geq \frac{d(u)}{|V|-1}$ holds for each vertex $u$ in $S$.
In this paper, we study the problem of finding a \emph{proportionally dense subgraph} (PDS) with a maximum number of vertices.
A \emph{proportionally dense subgraph} grants more importance to the vertices than the standard definition of a dense subgraph, as all the vertices in a PDS must be `satisfied', \textit{i.e.}\ respect the above condition.
This can be compared with \emph{defensive alliances} in graphs, where the vertices in the alliance must have at least as many neighbors inside the alliance than outside it \cite{rodriguez2009defensive,kristiansen04alliances}, without the notion of proportion of neighbors.

From a theoretical point of view, it is interesting to observe a problem that connects local and global properties of vertex subsets, interweaving the size of the subset and the number of neighbors.
This interesting paradigm has rarely been seen in graph theory problems.

The notion of proportionality of neighbors is closely related to community detection problems.
\citeauthor{Ols13} \cite{Ols13} defined a \emph{community structure} as a partition of the vertices of a graph into parts such that each vertex has a greater proportion of neighbors in its part than in any other part, each part being called a \emph{community}.
In the same paper, it was proved that any graph that is not a star contains a community structure that can be found in polynomial time (if there is no restriction on the number of communities), but that it is \NP-complete to decide if a given subset of vertices can belong to a same community of a community structure.
The special case where the community structure contains exactly two communities, namely a \emph{$2$-community structure}, has been studied in several classes of graphs: a $2$-community structure always exists and can be found in polynomial time in trees, graphs with maximum degree $3$, minimum degree $|V|-3$, and complements of bipartite graphs \cite{BCP17}.
Recently, the notion of $2$-community structure has been studied under the name of $2$-PDS partition \cite{BazChleDal2018}.
In this paper, the authors described an infinite family of graphs without a $2$-PDS partition, and a second infinite family of graphs without a connected $2$-PDS partition (but with a disconnected one).
These results answer some open questions originally introduced in \cite{BCP17}.
However, the complexity of finding a $2$-PDS partition remains unknown in general graphs, and for larger (fixed) number of PDS's.
As there is equivalence between \emph{proportionally dense subgraph} and \emph{community} (with regard to the above definition), one may interpret the problem of finding a proportionally dense subgraph of maximum size as finding a community of maximum size.
Hence, all the results presented in this paper can also be applied for community related problems.

\cref{prelim} introduces the basic notations used in the paper.
\cref{section:hardness} presents various hardness results of the {\MClong} problem.
\cref{section:approx} gives positive results about the approximation of this problem.
We prove that the the problem can be solved in linear time on Hamiltonian cubic graphs in \cref{section:hamiltonian}.
Conclusion and open problems are given in \cref{section:conclusion}.

\section{Preliminaries}\label{prelim}

Throughout the paper, we assume that all graphs are simple, undirected and connected.
For a graph $G=(V,E)$, we denote by $N(v)$ the set of neighbors of $v\in V$ and by $d(v)$ the degree of $v$, and thus $d(v)=|N(v)|$.
Also, $\Delta (G)$ denotes the maximum degree of $G$ (or simply $\Delta$ when no confusion arises).

In addition, given a subset of vertices $S \subset V$, we define $d_S(v)=|N(v)\cap S|$ and $\overbar{S}:=V\setminus S$;
also, $G[S]$ represents the induced subgraph of $S$ in $G$.
\medskip

A \emph{star} is a complete bipartite graph~$K_{1,\ell}$ for any $\ell\geq 1$.
A \emph{split graph} is a graph in which the vertices can be partitioned into an independent set and a clique.

\bigskip
\noindent\textbf{The Maximum Proportionally Dense Subgraph problem}
\begin{definition}\label{def:PDS}
Let $G=(V,E)$ be a graph and $S \subset V$, such that $2 \leq |S| < |V|$.
We say that the induced subgraph $G[S]$ is a \emph{proportionally dense subgraph} (PDS) if for each vertex $u \in S$,
    \begin{equation}\label{eq:vertex satisfied}
			\frac{d_S(u)}{|S|-1} \geq \frac{ d(u)}{|V|-1}\,,\\ %
			\text{which is equivalent to \ }\\
			\frac{d_S(u)}{|S|-1} \geq \frac{ d_{\overbar{S}}(u)}{|\overbar{S}|}\,.
		\end{equation}
We say a vertex $u$ is \emph{satisfied} (in $G[S]$) if it respects \cref{eq:vertex satisfied}.
The size of the proportionally dense subgraph $G[S]$ corresponds to the cardinality of $S$.
\end{definition}

The proof of the above equivalence from \cref{eq:vertex satisfied} can be found in \cite{BCP17}.

\medskip

\noindent\parbox{\linewidth}{
	\medskip
	\noindent \textsc{Max Proportionally Dense Subgraph} ({\MC})\\
	\noindent \textbf{Input:} A graph $G$.\\
	\noindent \textbf{Output:} A proportionally dense subgraph in $G$ of maximum size.\\
	\medskip
}

A proportionally dense subgraph may be connected or not.
We study both cases and talk about \emph{connected PDS} in the former case.
Notice that there exist graphs for which all proportionally dense subgraphs of maximum size are not connected, even if the graph is a cubic graph or a caterpillar.
In the cubic graph illustrated in \cref{fig:disconnected max com}, the gray vertices represent a PDS of size $7$, which is not connected.
In fact, any connected induced subgraph on the set $S$ with at least $6$ vertices contains at least one vertex $u$ of degree $1$ in $S$, which is not satisfied since $\frac{d_S(u)}{|S|-1} \leq \frac{1}{6-1} < \frac{2}{4} \leq \frac{d_{\overbar{S}}(u)}{|\overbar{S}|}$.
It can be checked that the maximum size for a PDS is $7$ but only $5$ for a connected PDS.
Similarly, in the caterpillar in \cref{fig:disconnected max com}, any connected induced subgraph of size at least $12$ has one vertex unsatisfied.
The maximum size for a PDS is $12$, while only $8$ for a connected PDS.

\begin{figure}[htpb!]
\begin{center}
\begin{tikzpicture}[scale=.7,rotate=0] %

\tikzset{every node/.style={draw,label distance = -3pt,circle,fill=white,inner sep=2.5pt}}
\tikzstyle{gray}=[draw,fill=gray!35,circle,inner sep=2.5pt]

	\begin{scope}
		\node (a) at (0,0) {};
		\node[gray] (b) at (1.5,0) {};
		\node[gray] (c) at (1.5,1) {};
		\node[gray] (d) at (1.5,-1) {};
		\draw (a) to (b);
		\draw (a) to (c);
		\draw (a) to (d);
		\draw (b) to (c);
		\draw (b) to (d);
		
		\node[gray] (e) at (3,0) {};
		\draw (c) to (e);
		\draw (d) to (e);
	
		\node (f) at (4,0) {};
		\draw (e) to (f);
		
		\node (g) at (5.5,-1) {};
		\node[gray] (h) at (5.5,1) {};
		\node[gray] (i) at (5.5,0) {};
		\node[gray] (j) at (7,0) {};
		\draw (j) to (g);
		\draw (j) to (i);
		\draw (j) to (h);
		\draw (i) to (g);
		\draw (i) to (h);
		
		\draw (h) to (f);
		\draw (g) to (f);
	\end{scope}

    \begin{scope}[xshift=9.5cm,yshift=-1cm]
            \node[gray] (a1) at (0,-.50) {};
\node[gray] (a2) at (0,.25) {};
\node[gray] (a3) at (0,1.75) {};
\node[gray] (a4) at (0,2.50) {};
\node[gray] (aa) at (0,1) {};

\node[gray,label={}] (a5) at (1,1) {};
\node[label={}] (ax) at (2,1) {};
\node[label={}] (z) at (3,1) {};
\node[label={}] (bx) at (4,1) {};
\node[gray,label={}] (b5) at (5,1) {};

\node[gray](b1) at (6,-.50) {};
\node[gray] (b2) at (6,.25) {};
\node[gray] (b3) at (6,1.75) {};
\node[gray] (b4) at (6,2.50) {};
\node[gray] (bb) at (6,1) {};

\path[draw,]  (a5) -- (ax) -- (z) -- (bx) -- (b5) {};
\draw[] (a1) -- (a5);
\draw[] (a2) -- (a5);
\draw[] (a3) -- (a5);
\draw[] (a4) -- (a5);
\draw[] (aa) -- (a5);

\draw[] (b1) -- (b5);
\draw[] (b2) -- (b5);
\draw[] (b3) -- (b5);
\draw[] (b4) -- (b5);
\draw[] (bb) -- (b5);

    \end{scope}

\end{tikzpicture}
 
	\end{center}
\caption{Two graphs in which all PDS of maximum size are not connected. Gray vertices represent a PDS of maximum size in each graph.}\label{fig:disconnected max com}
\end{figure}
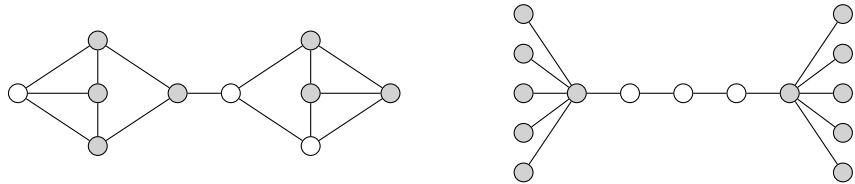

\noindent %

\section{Hardness results}\label{section:hardness}

In this section we prove several hardness results for {\MPDS} on split and bipartite graphs and further extend the results to prove that deciding if a PDS is inclusion-wise maximal is \textsf{co-NP}-complete.

We construct two polynomial-time reductions from {\MIS}, which is known to be \NP-hard \cite{Karp1972}.

\noindent\parbox{\linewidth}{
	\medskip
	\noindent {\MIS}\\
	\noindent \textbf{Input:} A graph $G$.\\
	\noindent \textbf{Output:} A subset of pairwise non-adjacent vertices in $G$ of maximum size.\\
}

\subsection{Split graphs}

We first describe a polynomial-time reduction, and then prove two intermediate results allowing us to easily prove the \NP-hardness of {\MPDS} on split graphs.

\begin{definition}\label[definition]{definition: split reduction}
	Let $G=(V,E)$ be a graph not isomorphic to a star.
	We define the construction $\sigma$ transforming the graph $G$ into $G':= \sigma(G)$, where $G'=(V',E')$ is defined as follows:
    \begin{compactitem}
		\item $V' := \{z_1,z_2\} \cup M \cup N$, where $N := V$, $M := \{ uv : \{u,v\} \in E\}$ and $z_1$, $z_2$ are two additional vertices;

		\item for each $e \in M$ and each $u \in N$, the edge $\{e,u\}\in E'$ if and only if $u \notin e$;

		\item the set $M \cup \{z_1,z_2\}$ induces a clique in $G'$.
	\end{compactitem}
\end{definition}

Obviously, the construction $\sigma$ can be done in polynomial time.
Notice that $G'$ is a split graph, and is connected if and only if $G$ is not isomorphic to a star.
See \cref{fig:transformation} for an example.

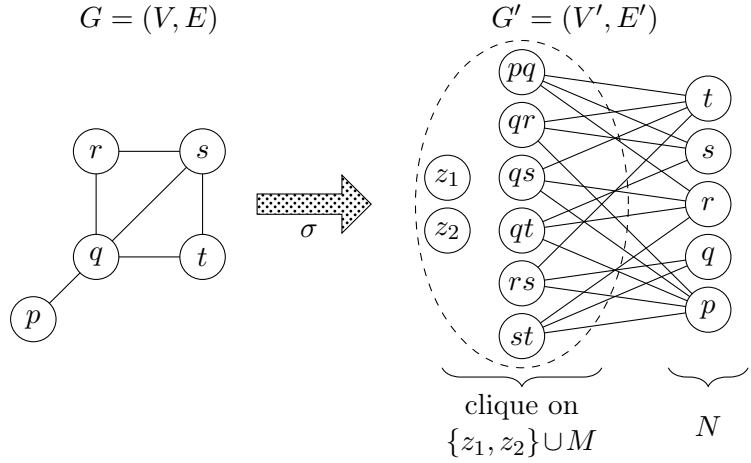
\begin{figure}[ht]
	\centering
\begin{tikzpicture}[scale=0.7]
	\begin{scope}[every node/.style={auto,circle,draw,minimum size=17pt,inner sep=0,outer sep=0}]
		\node (A) at (-.18,-.18) {$p$};
		\node[] (B) at (1,1) {$q$};
		\node[] (C) at (1,3) {$r$};
		\node[] (D) at (3,3) {$s$};
		\node (E) at (3,1) {$t$};
		
	\end{scope}
	\node[rectangle] (G) at (2,5.5) {$G=(V,E)$};

	\begin{scope}[every path/.style={}]
		\draw (A) -- (B);
		\draw (B) -- (C);
		\draw (B) -- (E);
		\draw (C) -- (D);
        \draw (D) -- (B);
		\draw (D) -- (E);
	\end{scope}

	\draw node[single arrow,draw=black,pattern=crosshatch dots,minimum height=1.5cm,label=below:$\sigma$] at (5,2) {};
	
	\begin{scope}[xshift=7.5cm,every node/.style={auto,circle,draw,minimum size=17pt,inner sep=0,outer sep=0}]
		\node (Z1) at (0.1,2.5) {$z_1$};
		\node (Z2) at (0.1,1.5) {$z_2$};
		
		\node (AB) at (1.5,4.5) {$pq$};
		\node (BC) at (1.5,3.5) {$qr$};
        \node (BD) at (1.5,2.5) {$qs$};
		\node (BE) at (1.5,1.5) {$qt$};
		\node (CD) at (1.5,0.5) {$rs$};
		\node (DE) at (1.5,-0.5) {$st$};

		\draw[dashed] (1.5,2) ellipse (2cm and 3.1cm);
		\draw [decorate,decoration={brace,amplitude=7pt}] (3,-1.1) -- (0,-1.1);
        \draw node[draw=none,rectangle, text width=2cm, text centered] at (1.5,-2.2) {clique on $\{z_1,z_2\} \cup M$};
        
        \draw [decorate,decoration={brace,amplitude=7pt}] (5.75,-1.1) -- (4.25,-1.1);
        \draw node[draw=none,rectangle, text centered] at (5,-2.2) {$N$};

		\node (E) at (5,4) {$t$};
		\node[] (D) at (5,3) {$s$};
		\node[] (C) at (5,2) {$r$};
		\node[] (B) at (5,1) {$q$};
		\node (A) at (5,0) {$p$};
	
		\begin{scope}[every path/.style={}]
			\draw (A) -- (BC);
			\draw (A) -- (BE);
			\draw (A) -- (CD);
            \draw (A) -- (BD);
			\draw (A) -- (DE);
			
			\draw (B) -- (CD);
			\draw (B) -- (DE);

			\draw (C) -- (AB);
			\draw (C) -- (BE);
			\draw (C) -- (DE);
			\draw (C) -- (BD);
			
			\draw (D) -- (AB);
			\draw (D) -- (BC);
			\draw (D) -- (BE);

			\draw (E) -- (AB);
			\draw (E) -- (BC);
			\draw (E) -- (CD);
            \draw (E) -- (BD);
		\end{scope}
	\node[draw=none,rectangle] (G') at (2.5,5.5) {$G'=(V',E')$};
	\end{scope}

\end{tikzpicture} 	\caption{The graph $G'$ obtained from the graph $G$ using the transformation $\sigma$.}\label{fig:transformation}
\end{figure}

\begin{lemma}\label[lemma]{lemma: split at least one non-neighbor}
 Let $G=(V,E)$ be a graph not isomorphic to a star and let $G'=(V',E')$ be such that $G' = \sigma(G)$.
 Let $S \subset V'$ be a set of vertices such that $M \cup \{z_1,z_2\} \subseteq S$.
	Then a vertex $e \in M$ is satisfied in $G'[S]$ if and only if $d_S(e) \geq |S|-2$.
\end{lemma}
\begin{proof}
    A vertex $e \in M$ has degree $d(e) = |V'| - 3$.
	Hence, if $d_S(e) < |S|-2$, then $d_{\overbar{S}}(e) = |\overbar{S}|$ and $e$ is not satisfied in $G'[S]$ as it does not respect \cref{eq:vertex satisfied}.
	However, if $d_S(e) \geq |S|-2$, then $d_{\overbar{S}}(e) < |\overbar{S}|$.
    Also, since $G$ is connected, $|M| \geq |N|-1$, and hence $|S| \geq |M|+2 > |N| \geq |\overbar{S}|$ and we have
	\[
	|\overbar{S}| \cdot d_S(e) \geq
	|\overbar{S}| \cdot (|S| - 2) \geq
	(|\overbar{S}| -1) \cdot (|S| - 1) \geq
	(|S| - 1) \cdot d_{\overbar{S}}(e) \,,
	\]
	and thus $e$ is satisfied in $G'[S]$.
\end{proof}

\begin{lemma}\label[lemma]{lemma: split one can find bigger}
	Let $G=(V,E)$ be a graph not isomorphic to a star and let $G'=(V',E')$ be such that $G' = \sigma(G)$.
	Let $S_1 \subset V'$ such that $G'[S_1]$ is a PDS.
	Then, there exists $S_2 \subset V'$ such that $G'[S_2]$ is a PDS, $|S_2| \geq |S_1|$ and $M\cup \{z_1,z_2\} \subseteq S_2$.
	Moreover, $S_2$ can be found in polynomial time.
\end{lemma}
\begin{proof}
    Firstly, we show that $N\nsubseteq S_1$.
	\begin{compactitem}
	\item if $S_1 = N$, since $G'[N]$ is an independent set, then any vertex $u\in S_1$ has $d_{S_1}(u)=0$ and $d_{\overbar{S}_1}(u) >0$; hence $u$ does not satisfy \cref{eq:vertex satisfied} and $G'[S_1]$ is not a PDS;

	\item if $N\subset S_1$, then $\overbar{S}_1$ is a subset of the clique $M \cup \{z_1,z_2\}$; it means any vertex $u \in S_1 \cap (M \cup \{z_1,z_2\})$ has $d_{\overbar{S}_1}(u) = |\overbar{S}_1|$  and $d_{S_1}(u) < |S_1|-2$, and thus
	$$ |\overbar{S}_1| \cdot d_{S_1}(u) < |\overbar{S}_1| \cdot (|S_1|-2) < |\overbar{S}_1| \cdot (|S_1|-1) = (|S_1|-1) \cdot d_{\overbar{S}_1}(u) \,,$$
	so $u$ does not satisfy \cref{eq:vertex satisfied} and $G'[S_1]$ is not a PDS.
	\end{compactitem}

   Now, let $S_2 := S_1 \cup M \cup \{z_1,z_2\}$ and $\overbar{S}_2 := V' \setminus S_2$.
   
   Observe that for any $f \in S_1 \cap M$, $d_{S_2}(f) - d_{S_1}(f) = |S_2| - |S_1| \geq 0$ and $d_{\overbar{S}_2}(f) \leq d_{\overbar{S}_1}(f)$.
   Thereby, we obtain $\frac{d_{S_2}(f)}{|S_2|-1} \geq \frac{d_{S_1}(f)}{|S_1|-1}\geq \frac{d_{\overbar{S}_1}(f)}{|S_1|}\geq \frac{d_{\overbar{S}_2}(f)}{|S_2|}$, so $f$ is satisfied in $S_2$.
   Also, if a vertex in $M$ is satisfied in $S_2$, then according to \cref{lemma: split at least one non-neighbor} it is also satisfied in any $S_2' \subseteq S_2$, as long as $M \cup \{z_1,z_2\} \subseteq S_2'$.
   
   If there exists $e \in M \setminus S_1$ which is not satisfied in $S_2$, then following \cref{lemma: split at least one non-neighbor} it holds $d_{S_2}(e) < |S_2|-2$.
   Thus, there exists a vertex $u \in S_2 \cap N$, non-adjacent to $e$, which we can transfer from $S_2$ to $\overbar{S}_2$.
   Obviously, at most $|M \setminus S_1|$ transfers are needed to satisfy all the vertices in $S_2$, and thus $|S_2| \geq |S_1|$ holds true.
   Since $S_2 \cap N \subseteq S_1 \cap N$ and $N \nsubseteq S_1$, then $S_2 \neq V'$.
   
   Note that $\overbar{S}_2 \subseteq N$ and that each vertex $u \in S_1 \cap N$ is satisfied in $S_2$, since $d_{\overbar{S}_2}(u) = 0$.
   Clearly, $z_1$ and $z_2$ are satisfied in $S_2$.
   Thus, $G'[S_2]$ is a PDS, and it can be found in polynomial time.
\end{proof}

Notice that \cref{lemma: split one can find bigger} implies that there exists a PDS of maximum size in $G'$ that is connected.
Hence, the following result also holds when looking for a connected PDS.

\begin{theorem}\label{theorem: max pds hard on split}
	{\MPDSlong} is \NP-hard on split graphs.
\end{theorem}
\begin{proof}
  Let $G=(V,E)$ be a graph not isomorphic to a star, $G' = (V',E')$ be such that $G' = \sigma(G)$, and $k \in \{1,\dots,|V|-2\}$.
  Notice that since $G$ is connected and not isomorphic to a star, then there is no independent set of size $|V|-1$ in $G$.
  We claim that there is an independent set of size at least $k$ in $G$ if and only if there is a PDS of size at least $|M|+2+k$ in $G'$.

  Let $R$ be an independent set of $G$ of size at least $k$.
  In $G'$, we define $S := M \cup \{z_1,z_2\} \cup R$ and $\overbar{S} := V '\setminus S$.
  First, note that $R \subseteq N$ thus $\overbar{S} = N \setminus R$.
  The vertices in $S \cap N \cup \{z_2,z_2\}$ are obviously satisfied in $G'[S]$ as they only have neighbors in $S$.
  Hence, if there exist unsatisfied vertices, then they must be from the set $M$.
  Choose a vertex $e \in M$.
  Since $R$ is an independent set of $G$, then for each edge $e=\{u,v\} \in E$  at most one of the vertices $u$ and $v$ belongs to $R$.
  Hence, the vertex $e \in M$ is not adjacent to at most one vertex in $S$, and thus $d_S(e) \geq |S|-2$.
  According to \cref{lemma: split at least one non-neighbor}, the vertex $e$ is satisfied in $G[S]$.
  Thus, $G[S]$ is a PDS of size at least $|M|+2+k$.

  Let $S\subset V'$ be of size at least $|M|+2+k$ such that $G'[S]$ is a PDS.
  According to \cref{lemma: split one can find bigger}, there exists $S' \subset V'$ such that $G'[S']$ is a PDS, $|S'| \geq |S|$ and $\{z_1,z_2\} \cup M \subseteq S'$.
  We claim that $R' := S' \cap N$ is an independent set of $G$ of size at least $k$.
  Obviously $|R'| \geq k$.
  Moreover, \cref{lemma: split at least one non-neighbor} states that for all satisfied vertices $e \in M$, $d_{S'}(e) \geq |S'|-2$.
  Hence, for each vertex $e \in M$ there is at most one vertex $u \in S'$ that is not adjacent to $e$.
  Since the vertices $e \in M$ and $u \in N$ are not adjacent in $G'$, it implies that $u \in e$ in $G$, and therefore the edge $e \in E$ has at most one endpoint $u \in R'$ in the graph $G$.
  Thus, $R'$ is an independent set of size at least $k$.
\end{proof}

\begin{proposition}\label{prop:apx-hard}
It is \NP-hard to approximate {\MPDSlong} within $1.0026028$ on split graphs, and hence the problem is \APX-hard (even on split graphs).
\end{proposition}
\begin{proof}
Let $I$ be an instance of {\MIS} on a cubic graph $G=(V,E)$.
It is known that it is \NP-hard to decide whether $opt(I)< \frac{12\tau+11+2\varepsilon}{24\tau+28} \cdot |V|$ or $opt(I)>\frac{12\tau+12-2\varepsilon}{24\tau+28} \cdot |V|$, for any $\varepsilon>0$, where $\tau\leq6.9$ \cite{chleb2006apx}.

We construct an instance $I'$ of {\MC} defined on the graph $G'=(V',E')$ such that $G' = \sigma(G)$.
Note that $M \subset V'$ is of size $|E|$, that is $|M| = |E| = \frac{3|V|}{2}$ since $G$ is cubic.
From \cref{theorem: max pds hard on split}, we know that $opt(I') = |M| + 2 + opt(I)$.
Consequently, it is \NP-hard to decide whether $opt(I') < |M| + 2 + \frac{12\tau+11+2\varepsilon}{24\tau+28} \cdot |V| = \frac{48\tau+53+2\varepsilon}{24\tau+28} \cdot |V| +2$ or $opt(I') > |M| + 2 + \frac{12\tau+12-2\varepsilon}{24\tau+28} = \frac{48\tau+54-2\varepsilon}{24\tau+28} \cdot |V|+2$.
We obtain that it is \NP-hard to approximate {\MC} within $1.0026028$.
\end{proof}

\subsection{Bipartite graphs}

In the following, we modify the previous construction in order to prove the \NP-hardness of {\MPDS} on bipartite graph.
The reduction will also be used to show the \NP-hardness of an ``extension version'' of the problem, implying the \textsf{co-NP}-completeness of deciding if a PDS is inclusion-wise maximal.

\begin{definition}\label[definition]{definition: bipartite global reduction}
	Let $G=(V,E)$ be a graph not isomorphic to a star, and an integer $k$ such that $1 \leq k < |V|-1$.
	We define the construction $\beta$ transforming the graph $G$ into $G':= \beta(G,k)$, where $G'=(V',E')$ is defined as follows:
	\begin{compactitem}
		\item $V' := L \cup M \cup N$, where $N := V$, $M := \{ uv : \{u,v\} \in E\}$ and $L$ contains $|L| := |M| \cdot (|V|-k-1)-k+1$ additional vertices;

		\item for each $e \in M$ and each $u \in N$, the edge $\{e,u\}\in E'$ if and only if $u \notin e$;

		\item for each $e \in M$ and each $v \in L$, the edge $\{e,v\} \in E'$.
	\end{compactitem}
\end{definition}

Obviously, the construction $\beta$ can be done in polynomial time.
Clearly, $G'$ is connected if and only if the input graph is not isomorphic to a star.
Also, notice that $G'$ is a bipartite graph as there are edges only between $M$ and $L \cup N$.
See \cref{fig:extension} for an example.

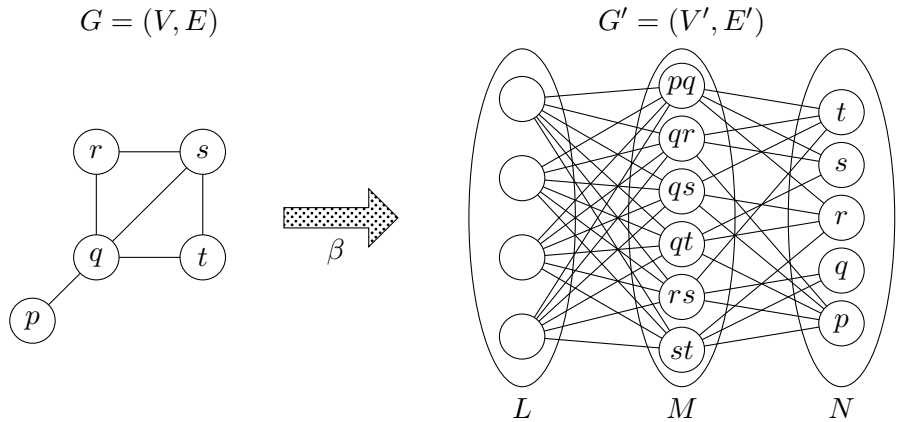
\begin{figure}[htb]
	\begin{center}
\begin{tikzpicture}[scale=0.7]
	\begin{scope}[yshift=-0.75cm,every node/.style={auto,circle,draw,minimum size=17pt,inner sep=0,outer sep=0}]
		\node (A) at (-6.2,-1.2) {$p$};
		\node[] (B) at (-5,0) {$q$};
		\node[] (C) at (-5,2) {$r$};
		\node[] (D) at (-3,2) {$s$};
		\node (E) at (-3,0) {$t$};
		
	\end{scope}
	\node[rectangle] (G) at (-4,3.7) {$G=(V,E)$};

	\begin{scope}[every path/.style={}]
		\draw (A) -- (B);
		\draw (B) -- (C);
		\draw (B) -- (E);
		\draw (C) -- (D);
		\draw (B) -- (D);
		\draw (D) -- (E);
	\end{scope}
	
	\draw node[single arrow,draw=black,pattern=crosshatch dots,minimum height=1.5cm,label=below:$\beta$] at (-0.5,0) {};

	\begin{scope}[xshift=3cm,every node/.style={auto,circle,draw,minimum size=17pt,inner sep=0,outer sep=0}]
    
	\draw (0,0) ellipse (1cm and 3.2cm);
   	\draw (3,0) ellipse (1cm and 3.2cm);
    \draw (6,0) ellipse (1cm and 3.2cm);

\node[black] (l1) at (0,2.25){};
\node[black] (l2) at (0,0.75){};
\node[black] (l3) at (0,-0.75){};
\node[black] (l4) at (0,-2.25){};

\node[black] (m1) at (3,2.5){$pq$};
\node[black] (m5) at (3,1.5){$qr$};
\node[black] (m6) at (3,0.5){$qs$};
\node[black] (m2) at (3,-0.5){$qt$};
\node[black] (m4) at (3,-1.5){$rs$};
\node[black] (m3) at (3,-2.5){$st$};

\node[black] (a) at (6,-2){$p$};
\node[black] (b) at (6,-1){$q$};
\node[black] (c) at (6,0){$r$};
\node[black] (d) at (6,1){$s$};
\node[black] (e) at (6,2){$t$};

\draw (l1) -- (m1);
\draw (l1) -- (m2);
\draw (l1) -- (m3);
\draw (l1) -- (m4);
\draw (l1) -- (m5);
\draw (l1) -- (m6);

\draw (l2) -- (m1);
\draw (l2) -- (m2);
\draw (l2) -- (m3);
\draw (l2) -- (m4);
\draw (l2) -- (m5);
\draw (l2) -- (m6);

\draw (l3) -- (m1);
\draw (l3) -- (m2);
\draw (l3) -- (m3);
\draw (l3) -- (m4);
\draw (l3) -- (m5);
\draw (l3) -- (m6);

\draw (l4) -- (m1);
\draw (l4) -- (m2);
\draw (l4) -- (m3);
\draw (l4) -- (m4);
\draw (l4) -- (m5);
\draw (l4) -- (m6);

\draw (c) -- (m1);
\draw (d) -- (m1);
\draw (e) -- (m1);

\draw (a) -- (m2);
\draw (c) -- (m2);
\draw (d) -- (m2);

\draw (a) -- (m3);
\draw (b) -- (m3);
\draw (c) -- (m3);

\draw (a) -- (m4);
\draw (b) -- (m4);
\draw (e) -- (m4);

\draw (a) -- (m5);
\draw (e) -- (m5);
\draw (d) -- (m5);

\draw (a) -- (m6);
\draw (c) -- (m6);
\draw (e) -- (m6);
        
    \node[draw=none,rectangle] (L) at (0,-3.6) {$L$};
	\node[draw=none,rectangle] (M) at (3,-3.6) {$M$};
	\node[draw=none,rectangle] (N) at (6,-3.6) {$N$};

	\node[draw=none,rectangle] (G') at (3,3.7) {$G'=(V',E')$};

		\end{scope}

		\end{tikzpicture}
 		\caption{The graph $G'$ obtained from $G$ using  the transformation $\beta$ and $k=3$.}\label{fig:extension}
	\end{center}
\end{figure}

We now prove intermediate results, which help concluding that {\MC} is \NP-complete on bipartite graphs.

\begin{lemma}\label[lemma]{lemma: l is well chosen}
	Let $m$, $n$ and $k$ be integers such that $1\leq k <n-1 \leq m$ and $\ell := m \cdot (n-k-1) -k+1$.
	Then $\frac{\ell+k-1}{\ell+m+k-1} = \frac{n-k-1}{n-k}$.
\end{lemma}
\begin{proof}
	$
	(n-k) \cdot (\ell+k-1) = (n-k-1) \cdot (\ell+k-1) +\ell +k-1
	= (n-k-1) \cdot (\ell+k+1) + m \cdot (n-k-1)
	= (n-k-1) \cdot (\ell+m+k+1)\,.
	$
\end{proof}

\begin{lemma}\label[lemma]{lemma: bipartite at least one non-neighbor}
	Let $G=(V,E)$ be a graph not isomorphic to a star, $k$ an integer such that $1 \leq k < |V|-1$ and $G'=(V',E')$ be such that $G' = \beta(G,k)$.
	Let $S \subset V'$ be such that $|S| \geq |L| + |M| + k$.
	Then a vertex $f \in M$ is satisfied in $G'[S]$ if and only if $d_{\overbar{S}}(f) < |\overbar{S}|$.
\end{lemma}
\begin{proof}
	If $d_{\overbar{S}}(f) = |\overbar{S}|$, $f$ is obviously not satisfied.
	If $d_{\overbar{S}}(f) < |\overbar{S}|$, then notice that $d(f) = |L|+|N|-2 = |V'| - |M| - 2$.
    Therefore, $d_S(f) = d(f) - d_{\overbar{S}}(f) \geq |V'| - |M| - 2 - |\overbar{S}| + 1 = |S| - |M| - 1$.
    Also, $|\overbar{S}| \leq |N|-k$.
	Consequently, according to \cref{lemma: l is well chosen},
	\[
		\frac{d_S(f)}{|S|-1} = \frac{|S| -|M|-1}{|S|-1} \geq
        \frac{|L| + k -1}{|L|+|M|+k-1} = \frac{|N|-k-1}{|N|-k} \geq \frac{d_{\overbar{S}}(f)}{|\overbar{S}|}\,.
	\]

\end{proof}

\begin{lemma}\label[lemma]{lemma: bipartite one can find bigger}
	Let $G=(V,E)$ be a graph not isomorphic to a star, $k$ an integer, $1 \leq k < |V|-1$, and let $G'=(V',E')$ be such that $G' = \beta(G,k)$.
	Let $S_1 \subset V'$ such that $G'[S_1]$ is a PDS and $|S_1| \geq |L|+|M|+k$.
	Then, there exists $S_2 \subset V'$ such that $G'[S_2]$ is a PDS, $|S_2| \geq |S_1|$ and $L \cup M \subseteq S_2$.
	Moreover, $S_2$ can be found in polynomial time.
\end{lemma}
\begin{proof}
 	First, we prove that $M \subset S_1$.
    As $|S_1| \geq |L|+|M|+k > |M|+|N|$, then $S_1 \cap L \neq \emptyset$.
    Take a vertex $z \in S_1 \cap L$ and notice that since $d(z) = |M|$, then $d_{\overbar{S}_1}(z) = |M \setminus S_1|$.
    The vertex $z$ is satisfied in $G'[S_1]$ if and only if
    \[
    	\frac{|M| - d_{\overbar{S}_1}(z)}{|L|+|M|+k-1} \geq \frac{d_{S_1}(z)}{|S_1|-1} \geq \frac{d_{\overbar{S}_1}(z)}{|\overbar{S}_1|} \geq \frac{d_{\overbar{S}_1}(z)}{|N|-k} \,.
    \]
    This implies that
		\begin{flalign*}
					|M| \cdot (|N|-k) - d_{\overbar{S}_1}(z) \cdot (|N|-k) &\geq d_{\overbar{S}_1}(z) \cdot (|L|+|M|+k-1)\\
			\iff	|M| \cdot (|N|-k) - d_{\overbar{S}_1}(z) \cdot (|N|-k) &\geq d_{\overbar{S}_1}(z) \cdot |M| \cdot (|N|-k)\\
			\iff 	|M| \cdot (|N|-k) &\geq  d_{\overbar{S}_1}(z) \cdot (|M|+1) \cdot (|N|-k)\\
			\iff	0 &\geq d_{\overbar{S}_1}(z)\,.
		\end{flalign*}
	Thus, we have $d_{\overbar{S}_1}(z) = 0$ and conclude that $M \subset S_1$.
	
	Let $S_2 := S_1 \cup L \cup M$ and $f \in M$.
    As $f$ is satisfied in $G'[S_1]$, according to \cref{lemma: bipartite at least one non-neighbor}, we have $d_{\overbar{S}_1}(f) < |\overbar{S}_1|$.
    Since $f$ is connected to all the vertices in $L$, necessarily $d_{\overbar{S}_2}(f) < |\overbar{S}_2|$ and $f$ remains satisfied in $G'[S_2]$.
	Obviously, the vertices in $L$ are satisfied in $G'[S_2]$ since all their neighbours are in $M$.
	This is also true for the vertices in $N \cap S_2$.
\end{proof}

Notice that \cref{lemma: bipartite one can find bigger} implies that there exists a PDS of maximum size that is connected in $G'$.
Hence, the following result also holds when looking for a connected PDS.

\begin{theorem}\label[theorem]{theorem: NP hard on bipartite}
	{\MPDSlong} is \NP-hard on bipartite graphs.
\end{theorem}
\begin{proof}
	Let $G=(V,E)$ be a graph not isomorphic to a star, $k \in \{1,\dots,|V|-1\}$.
	Notice that since $G$ is connected and not isomorphic to a star, then there is no independent set of size $|V|-1$ in $G$.
	Let $G' = (V',E')$ such that $ G' = \beta(G,k)$.
	We claim that there is an independent set of size at least $k$ in $G$ if and only if there is a PDS of size at least $|L|+|M|+k$ in $G'$.

	Let $R$ be an independent set of $G$ of size at least $k$.
	In $G'$, we define $S := L \cup M \cup R$ and $\overbar{S} := V '\setminus S$.
	First, note that $R \subseteq N$ thus $\overbar{S} = N \setminus R$.
	The vertices in $L \cup R$ are obviously satisfied in $G'[S]$ as all their neighbors are in $S$.
	Hence, if there exists vertices not satisfied in $G'[S]$, then they must belong to the set $M$.
	Consider a vertex $e \in M$.
	Since $R$ is an independent set of $G$, then for each edge $e=\{u,v\} \in E$ at most one of the vertices $u$ and $v$ belongs to $R$, and, therefore, at least one belong to $\overbar{S}$.
	Therefore, the vertex $e \in M$ is not adjacent to at least one vertex in $\overbar{S}$, and thus $d_{\overbar{S}}(f) < |\overbar{S}|$.
	According to \cref{lemma: bipartite at least one non-neighbor}, $e$ is satisfied in $G[S]$.
	Thus, $G[S]$ is a PDS of size at least $|L|+|M|+k$.

	Let $S\subset V'$ be of size at least $|L|+|M|+k$ such that $G'[S]$ is a PDS.
	According to \cref{lemma: bipartite one can find bigger}, there exists $S' \subset V'$ such that $G'[S']$ is a PDS, $|S'| \geq |S|$ and $L \cup M \subseteq S'$.
	We claim that $R' := S' \cap N$ is an independent set of $G$ of size at least $k$.
	Obviously $|R'| \geq k$.
	\Cref{lemma: bipartite at least one non-neighbor} states that for all satisfied vertices $e \in M$, $d_{\overbar{S}'}(e) < |\overbar{S}'|$.
	Therefore, as $d_N(e) = |N|-2$ and $\overbar{S}' \subseteq N$, there is at most one vertex $u \in S' \cap N$ not adjacent to $e$.
  From the construction $\sigma$, if there is no edge between the vertices $e \in M$ and $u \in N$ in $G'$, then $u \in e$ in $G$.
  Hence, the edge $e \in E$ in $G$ has at most one vertex $u \in R'$.
	Thus, $R'$ is an independent set of size at least $k$.
\end{proof}

Below, we prove that deciding if a subset of vertices can be extended into a larger subset which induces a PDS is \NP-complete.
We obtain as a corollary that deciding if a PDS is inclusion-wise maximal is \textsf{co-NP}-complete.

\noindent\parbox{\linewidth}{
\medskip
\noindent	\textsc{PDS Extension}\\
\noindent	\textbf{Input:} A graph $G=(V, E)$, $U \subset V$.\\
\noindent	\textbf{Question:} Is there a vertex subset
$S \subset V$ such that $U \subset S$ and $G[S]$ is a proportionally dense subgraph?
\medskip
}

To prove that \textsc{PDS Extension} is \NP-complete, we use again the construction $\beta$ as defined in \cref{definition: bipartite global reduction}.

\begin{lemma}\label[lemma]{lemma: L and M means max PDS}
Let $G=(V,E)$ be a graph not isomorphic to a star, $k$ an integer, $1 \leq k < |V|-1$, and $G'=(V',E')$ be such that $G' = \beta(G,k)$.
Let $S \subset V'$ be such that $L \cup M \subset S$ and $G'[S]$ is a PDS.

Then $|S| \geq |L| + |M| + k$.
\end{lemma}
\begin{proof}
Let $u \in S \cap N$, and notice that $d_S(u) < |M|$, so there exists a vertex in $M$ which is not connected to $u$.
Let $f \in M$ be such a vertex.
Note that $d_S(f) \leq |S| - |M| -1$ and $d_{\overbar{S}}(f) \geq |\overbar{S}|-1$, as $f$ is not connected to $u$.

Let $k' := |N \setminus \overbar{S}| = |N| - |\overbar{S}|$.
We claim that $k' \geq k$.
Suppose by contradiction that $k' < k$.
Then $\frac{|L|+k'-1}{|L|+|M|+k'-1} < \frac{|L|+k-1}{|L|+|M|+k-1}$ and $\frac{|N|-k-1}{|N|-k} < \frac{|N|-k'-1}{|N|-k'}$.
According to \cref{lemma: l is well chosen}, we conclude that $\frac{|L|+k'-1}{|L|+|M|+k'-1} < \frac{|N|-k'-1}{|N|-k'}$.
Therefore,
\[
\frac{d_S(f)}{|S|-1} \leq \frac{|L|+k'-1}{|L|+|M|+k'-1} < \frac{|N|-k'-1}{|N|-k'} \leq \frac{d_{\overbar{S}}(f)}{|\overbar{S}|}\,,
\]
which contradicts that $f$ is satisfied, and thus that $G'[S]$ is a PDS.
We conclude that $|S| = |L| + |M| + k' \geq |L|+|M|+k$.
\end{proof}

\begin{theorem}\label[theorem]{theorem: PDS extension}
\textsc{PDS Extension} is \NP-complete on bipartite graphs.
\end{theorem}
\begin{proof}
Obviously, \textsc{PDS Extension} is in \NP.
Let $G=(V,E)$ be a graph not isomorphic to a star, $k \in \{1,\dots,|V|-1\}$.
Notice that since $G$ is connected and not isomorphic to a star, then there is no independent set of size $|V|-1$ in $G$. %
Let $G' = (V',E')$ such that $ G' = \beta(G,k)$.
We claim that there is an independent set of size at least $k$ in $G$ if and only if there is PDS of size of size at least $|L| + |M| + k$ in $G'$.

Assume there exists an independent set of size $k$ in $G$.
Then, there exists $S \subset V'$ of size $|S| \geq |L|+|M|+k$ such that $G'[S]$ is a PDS, and $L \cup M \subset S$ (see proof of \cref{theorem: NP hard on bipartite}).

According to \cref{lemma: L and M means max PDS}, if there exists $S \subset V'$ such that $G'[S]$ is a PDS and $L \cup M \subset S$, then $|S| \geq |L|+|M|+k$.
Therefore, there exists an independent set of size at least $k$ in $G$ (see proof of \cref{theorem: NP hard on bipartite}).

We conclude that deciding if there exists $S \subset V'$ such that $L \cup M \subset S$ and $G'[S]$ is a PDS is \NP-complete, and thus that \textsc{PDS Extension} is \NP-complete on bipartite graphs.
\end{proof}

Notice that the set $L \cup M$ is connected, and thus if it can be extended into a PDS, then the PDS is connected.
Hence, it is \NP-complete to decide whether a vertex subset (inducing a connected subgraph) can be extended into a connected PDS.
Furthermore, the set $L \cup M$ can induce a PDS or not, depending on the values of $k$ and $|V|$.
Indeed, $G'[L \cup M]$ is a PDS if and only if $\frac{|L|}{|L|+|M|-1} \geq \frac{|N|-2}{|N|}$, which implies $k \leq \frac{n}{2}$.
Therefore, we conclude that deciding if a PDS is inclusion-wise maximal is \textsf{co-NP}-complete.

\begin{corollary}
Let $G=(V,E)$ be a graph and $S \subset V$ such that $G[S]$ a proportionally dense subgraph.
Deciding if $S$ is inclusion-wise maximal is \textsf{co-NP}-complete on bipartite graphs.
\end{corollary}

\section{Approximation}\label{section:approx}

In this section we show that there exists a polynomial-time 2-approximation algorithm for {\MClong}, which establishes the \APX-completeness of the problem.
When the maximum degree $\Delta$ of the graph is bounded, the approximation ratio can be further improved to $(2-\frac{2}{\Delta +1})$ using a better upper bound on the size of a PDS.

\begin{lemma}\label[lemma]{lemma:not PDS means din leq dout}
Let $G=(V,E)$ be a graph  and  $S \subset V$ such that $G[S]$ is not a proportional dense subgraph.
If $|S|=\lceil\frac {|V|}2\rceil$, then there exists $u\in S$ such that $d_S(u) < d_{\overbar{S}}(u)$.
Moreover, if $|V|$ is even and $|S|=\frac{|V|}{2}+1$, then there exists $u\in S$ such that $d_S(u) \leq d_{\overbar{S}}(u)$.
\end{lemma}
\begin{proof}
	Let $S \subset V$ be a subset such that $G[S]$ is not a PDS.
	Then, there exists a vertex $u \in S$ such that \cref{eq:vertex satisfied} is not satisfied in $G[S]$, and therefore $|\overbar{S}| \cdot d_S(u) < (|S|-1) \cdot d_{\overbar{S}}(u)~(*)$.
 	\begin{compactitem}
		\item If $|S| = \lceil \frac{|V|}{2} \rceil$, the inequality $(*)$ implies $\lfloor \frac{|V|}{2} \rfloor \cdot d_S(u) < (\lceil \frac{|V|}{2} \rceil - 1) \cdot d_{\overbar{S}}(u) \leq \lfloor \frac{|V|}{2} \rfloor \cdot d_{\overbar{S}}(u)$, and hence $d_S(u) < d_{\overbar{S}}(u)$.
		
        \item If $|S| = \frac{|V|}{2} + 1$ ($|V|$ even), assume by contradiction that for each vertex  $v\in S$ it holds $d_S(v) > d_{\overbar{S}}(v)$.
		In particular, the inequality $(*)$ implies  $(\frac{|V|}{2}-1) \cdot (d_{\overbar{S}}(u)+1) < \frac{|V|}{2} \cdot d_{\overbar{S}}(u)$, 
		which is true if and only if $d_{\overbar{S}}(u) \geq \frac{|V|}{2}$.
		Thus, $d(u) = d_S(u) + d_{\overbar{S}}(u) > |V|-1$, a contradiction.
	\end{compactitem}
\end{proof}

\begin{theorem}\label{theorem:algo2approx}
    For any graph $G=(V,E)$, a proportionally dense subgraph of size $\lceil \frac{|V|}{2} \rceil$ or $\lceil \frac{|V|}{2} \rceil +1$ can be constructed in $\O(|V| \cdot |E|)$ time.
\end{theorem}
\begin{proof}
First, we show that \cref{algo:PDS at least n/2} terminates and returns a PDS of size $\lceil \frac{|V|}{2} \rceil$ or $\lceil \frac{|V|}{2} \rceil +1$.

\begin{algorithm}[ht]
	\caption{Find a proportional dense subgraph of size $\ceil{ \frac{|V|}{2} }$ or $\ceil{ \frac{|V|}{2} }+1$.}\label{algo:PDS at least n/2}
	\KwIn{$G=(V,E)$ a graph.}
	\KwOut{$S \subset V$  such that $G[S]$ is a PDS.}
	Let $S \subset V$ with $|S| = \lceil \frac{|V|}{2} \rceil$\;\label{algoline:define S for the first time}
	
	\While{$G[S]$ is not a PDS}{\label{algoline:check if PDS}

		Let $u \in S$ such that $d_{\overbar{S}}(u) - d_S(u)$ is maximum\;\label{algoline:chose u with max diff}
		$S := \overbar{S} \cup \{u\}$\;\label{algoline:update S by transferring u}
	}
	\Return $S$\;
\end{algorithm}

	\begin{compactitem}
	\item \textbf{Case 1: $|V|$ is odd.}
	Notice that at the end of each loop, the set $S$ is modified without changing its size $|S| = \frac{|V|+1}{2}=\lceil \frac{|V|}{2} \rceil$.
	If $G[S]$ is not a PDS, then according to \cref{lemma:not PDS means din leq dout} there exists an unsatisfied vertex $v\in S$ for which $d_S(v) < d_{\overbar{S}}(v)$.
	Therefore, the vertex $u$ chosen within the loop has the property $d_{\overbar{S}}(u) - d_S(u) > 0$.
	Thus, the size of the cut between $S$ and $\overbar{S}$ decreases after each loop and the algorithm terminates.
	\item \textbf{Case 2: $|V|$ is even.}
	Notice that \cref{algo:PDS at least n/2} starts with $|S| = \frac{|V|}{2}$.
	If $G[S]$ is not a PDS, then due to \cref{lemma:not PDS means din leq dout}, there exists a vertex $v\in S$ such that $d_S(v) < d_{\overbar{S}}(v)$.
	The selection of the vertex $u\in S$ inside the loop ensures that the size of the cut between $S$ and $\overbar{S}$ strictly decreases at the end of the loop.
	Now, observe that after the first loop, $|S| = \frac{|V|}{2} + 1$.
	If $G[S]$ is not a PDS, according to \cref{lemma:not PDS means din leq dout}, there exists a vertex $v\in S$ such that $d_S(v) \leq d_{\overbar{S}}(v)$.
	Therefore, the vertex $u$ inside the loop has $d_S(u) \leq d_{\overbar{S}}(u)$.
	Obviously, after the second loop, $|S| = \frac{|V|}{2}$.
	Since after each loop $|S|$ alternates between $\frac{|V|}{2}$ and $\frac{|V|}{2}+1$, the cut between $S$ and $\overbar{S}$ strictly decreases every two loops, and the algorithm terminates.
	\end{compactitem}
	
	It is easy to see that the while-loop is called at most $\O(|E|)$ times.
    Now, we prove how one can obtain a $\O(|V| \cdot |E|)$ running time by computing \cref{algoline:check if PDS,algoline:chose u with max diff,algoline:update S by transferring u} in $\O(|V|)$ time.
    
    \paragraph{Preprocessing}
    Once $S$ has been defined at \cref{algoline:define S for the first time}, compute and store the following \emph{properties} for each vertex $u \in V$: $d_S(u)$, $d_{\overbar{S}}(u)$, and whether $u$ belongs to $S$ or $\overbar{S}$.
    The computation of these properties for all the vertices can be done in $\O(|E|)$ time.
    While computing the properties, one can also choose a vertex $u \in S$ that maximises $d_{\overbar{S}}(u) - d_S(u)$ (as in \cref{algoline:chose u with max diff}).

    \paragraph{Main loop}
    If $d_{\overbar{S}}(u) - d_S(u) > 0$, then $S$ is not a PDS.
    However, if $d_{\overbar{S}}(u) - d_S(u) = 0$, then $S$ is a PDS if and only if $|S| < \frac{|V|}{2}+1$ (so we decide \cref{algoline:check if PDS} in constant time).
    Therefore, if $S$ is not a $PDS$, set $S := \overbar{S} \cup \{u\}$ (as in \cref{algoline:update S by transferring u}), update the properties of all the vertices and select $u \in S$ maximising $d_{\overbar{S}}(u) - d_S(u)$ (as in \cref{algoline:chose u with max diff}) in $\O(|V|)$.
    Then, repeat from the beginning of the main loop.
\end{proof}

\begin{corollary}\label[corollary]{cor:two-app}
{\MClong} is polynomial-time $2$-approximable.
\end{corollary}
\begin{proof}
For any graph $G=(V, E)$, \cref{algo:PDS at least n/2} yields a PDS of size at least $\ceil{\frac{|V|}{2}}$ and since any PDS has size at most $|V|-1$, we obtain a $2$-approximation algorithm.
\end{proof}

We proved the \APX-hardness of {\MC} in \cref{prop:apx-hard}, and hence we conclude the \APX-completeness of the problem.
\begin{corollary}
{\MClong} is \APX-complete.
\end{corollary}

In the following we show how the approximation ratio can be improved with regard to the maximum degree $\Delta$ of the graph.

\begin{lemma}\label[lemma]{lemma:upperbound-PDS}
Let $G=(V,E)$ be a graph and $S \subset V$ such that $G[S]$ is a proportionally dense subgraph.
Then $|S|\leq \lfloor \frac{|V|\cdot (\Delta (G) -1)+1}{\Delta (G)}\rfloor$.
\end{lemma}
\begin{proof}
Let $v$ be a vertex of $S$ with at least one neighbor in $\overbar{S} = V\setminus S$ (such a vertex exists since $G$ is connected).
Since $G[S]$ is a PDS, $v$ fulfills the proportion condition, that is $\frac{\Delta (G) -1}{|S|-1} \geq \frac{d_{S} (v)}{|S|-1} \geq \frac{d_{\overbar{S}}(v)}{|\overbar{S}|} \geq \frac{1}{|V|-|S|} $ which implies that $|S|\leq \frac{|V|\cdot (\Delta (G) -1)+1}{\Delta (G)}$, and hence $|S|\leq \lfloor \frac{|V|\cdot (\Delta (G) -1)+1}{\Delta (G)} \rfloor$.
\end{proof}

\begin{proposition}\label[proposition]{prop:upperbound-degree}
	\textsc{Max Proportionally Dense Subgraph} is polynomial-time $(2-\frac{2}{\Delta+1})$-approximable.
\end{proposition}
\begin{proof}
Let $G=(V,E)$ be a graph, $S$ be a set returned by \cref{algo:PDS at least n/2} and $opt(G)$ denote the size of a PDS of maximum size in $G$.
According to \cref{lemma:upperbound-PDS} we have $opt(G)\leq \frac{|V|\cdot (\Delta (G) -1)+1}{\Delta (G)}$.
Therefore, since $|V| \geq \Delta+1$ and $opt(G) \geq |S|$, we obtain
\begin{flalign*}
\frac{opt(G)}{|S|} \leq \frac{2\cdot opt(G)}{|V|} &\leq \frac{2 \cdot (|V| \cdot (\Delta-1) +1)}{|V| \cdot \Delta}\\
&\leq \frac{2 \cdot ((\Delta+1) \cdot (\Delta-1) +1)}{(\Delta+1) \cdot \Delta} = 2 - \frac{2}{\Delta+1}\,.
\end{flalign*}
\end{proof}

\Cref{algo:PDS at least n/2} shows that the decision version associated with {\MC} is in {\FPT} when parameterized by its natural parameter $k$ (\textit{i.e.}\ the size of a PDS).
  Indeed, if the parameter $k \leq \lceil \frac{|V|}{2} \rceil$, then a PDS of size greater than $k$ can be found in polynomial time using \cref{algo:PDS at least n/2}.
  On the other hand, if $k > \lceil \frac{|V|}{2} \rceil$, then we have $|V| < 2k$ and an exhaustive search can be done in $O(2^{2k})$ operations. %

	\def\u{<\!u\!>}

	\section{Hamiltonian cubic graphs}\label{section:hamiltonian}
	In this section we prove that all Hamiltonian cubic graphs of order $n$, except two graphs (see \cref{fig:counterexample-hamiltonian}), have a proportionally dense subgraph of the maximum possible size $\floor{ \frac{2n+1}{3} }$ (see \cref{lemma:upperbound-PDS} for an upper bound on a PDS size).
	Furthermore, we show that such a PDS can be found in linear time if a Hamiltonian cycle is given in the input.
	Note that almost all cubic graphs are Hamiltonian, as proved in \cite{robinson1992almost}.
	
	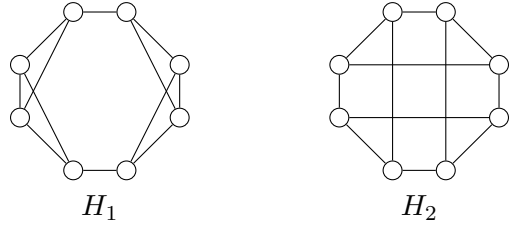
\begin{figure}[htpb!]
		\begin{center}
\begin{tikzpicture}[scale=.7,rotate=0] %

\tikzset{every node/.style={draw,label distance = -3pt,circle,fill=white,inner sep=2.5pt}}
	\begin{scope}
    	\node[draw=none,fill=none,rectangle] (h1) at (0.5,-3.7) {$H_1$};
		\node (7) at (0,0) {};
		\node (0) at (1,0) {};
		\node (6) at (-1,-1) {};
		\node (1) at (2,-1) {};
		\node (5) at (-1,-2) {};
		\node (2) at (2,-2) {};
		\node (4) at (0,-3) {};
		\node (3) at (1,-3) {};
		\draw (7) to (0);
		\draw (1) to (0);
		\draw (1) to (2);
		\draw (3) to (2);
		\draw (3) to (4);
		\draw (5) to (4);
		\draw (5) to (6);
		\draw (6) to (7);
		\draw (5) to (7);
		\draw (6) to (4);
		\draw (0) to (2);
		\draw (1) to (3);
		
        \node[draw=none,fill=none,rectangle] (h2) at (6.5,-3.7) {$H_2$};
        \node (17) at (6,0) {};
		\node (10) at (7,0) {};
		\node (16) at (5,-1) {};
		\node (11) at (8,-1) {};
		\node (15) at (5,-2) {};
		\node (12) at (8,-2) {};
		\node (14) at (6,-3) {};
		\node (13) at (7,-3) {};
		\draw (17) to (10);
		\draw (11) to (10);
		\draw (11) to (12);
		\draw (13) to (12);
		\draw (13) to (14);
		\draw (15) to (14);
		\draw (15) to (16);
		\draw (16) to (17);
		\draw (17) to (14);
		\draw (10) to (13);
		\draw (15) to (12);
		\draw (16) to (11);

	\end{scope}
	
\end{tikzpicture}
 		\end{center}
        \caption{Two Hamiltonian cubic graphs with $8$ vertices without PDS of size $\floor{\frac{2 \times 8+1}{3}}=5$.}\label{fig:counterexample-hamiltonian}
	\end{figure}

We represent a Hamiltonian cubic graph of order $n$ as a cycle with the vertices labeled in such a way that $(0,1,\dots,n-1)$ is a Hamiltonian cycle and a set of edges between non-successive vertices in the Hamiltonian cycle.
We always refer to this cycle when we say \emph{the} Hamiltonian cycle of a graph.
To avoid tedious notations, we use $i \in \mathbb{N}$ (with $0\in \mathbb{N}$) to refer to the vertex labeled by $i \mmod n$.

 \begin{definition}
	Let $G=(V,E)$ be a Hamiltonian cubic graph, $u \in V$.
    Let $P$ be a set of successive vertices in the Hamiltonian cycle labeled with $u$, $u+1$, \dots, $u-k-1$, with $k$ such that $|V| - 2 \geq k \geq 2$.
    The set $P$ is called a \emph{shift} if the first and the last vertices of the sequence, $u$ and $u-k-1$, are such that $d_P(u)=d_P(u-k-1)=2$.
 \end{definition}

   Notice that a shift $P$ contains $|V|-k$ vertices.
   Also, any vertex of $P$ has at least two neighbors in $P$.
   Consequently, if $k \geq \ceil{ \frac{|V|-1}{3} }$, then $|P| \leq \floor{\frac{2\cdot|V|+1}{3}}$, and the following holds for any $u \in P$:
   \[
        \frac{d_P(u)}{|P|-1} \geq \frac{2}{|V|-k-1} \geq \frac{1}{k} \geq \frac{d_{\overbar{P}}(u)}{|\overbar{P}|}\,.
   \]
   Thus, $G[P]$ is a PDS.
    If $k =  \ceil{ \frac{|V|-1}{3} }$, then $G[P]$ is a PDS of the maximum possible size $\floor{\frac{2\cdot|V|+1}{3}}$ (see \cref{lemma:upperbound-PDS}) and we call $P$ a \emph{good shift}.
    On the other hand, if $k=\ceil{ \frac{|V|-1}{3} }-1$, then the size of $P$ is one vertex larger than the size of the maximum possible PDS, and thus $G[P]$ is not a PDS.
    Such a shift is called an \emph{almost good shift}.

	In the following, we prove that either $G$ contains a good shift or we can find an almost good shift $P$ and a vertex $v\in P$ such that $G[P\setminus \{v\}]$ is a proportionally dense subgraph of the maximum possible size $\floor{\frac{2\cdot|V|+1}{3}}$.

	\begin{definition}
		Let $G=(V,E)$ be a Hamiltonian cubic graph.
        For each $v \in V$, we denote by $c(v)$ the non-successive neighbor of $v$ in the Hamiltonian cycle.
				Additionally, we define the subsets of vertices $L$ and $R$ in the following way for $k:=\ceil{ \frac{|V|-1}{3} }$:
		\begin{compactitem}
			\item $L:=\{u\in V : c(u) \in \{u-k, u-k+1,,\dots,u-2\}\}$;
			\item $R:=\{u\in V : c(u) \in \{u+2,u+3,\dots,u+k\}\}$.
		\end{compactitem}
	\end{definition}

	For a Hamiltonian cubic graph $G=(V,E)$ and $u \in V$, notice that $u \in L$ if and only if $c(u) \in R$, and symmetrically $u \in R$ if and only if $c(u) \in L$.
	This particularly implies that $|L| = |R| \leq \frac{|V|}{2}$.
	Moreover, notice that for a vertex $u\in L$, the set $P:=\{u,u+1,\dots , u-k-1\}$ cannot be a good shift, since $d_P(u)= 1$.
	In the same way, if $u\in R$, the set $P:=\{u+k+1,u+k+2,\dots ,u-1,u\}$ cannot be a good shift, since $d_P(u)= 1$.
	These observations are summed up in the following lemma.

	\begin{lemma}\label[lemma]{lemma:if not L then neighbor is R}
		Let $G=(V,E)$ be a Hamiltonian cubic graph, $k := \ceil{\frac{|V|-1}{3}}$ and $u \in V$.
		If $u \notin L$ and $(u-(k+1)) \notin R$, then the set $\{u,u+1,\dots,u-(k+1)-1,u-(k+1)\}$ is a good shift.
		Symmetrically, if $u \notin R$ and $(u+k+1) \notin L$, then the set $\{u+k+1,u+k+2,\dots ,u-1,u\}$ is a good shift.
	\end{lemma}
	\begin{proof}
		The proof is straightforward.
		Since $u \notin L$ and $(u-(k+1)) \notin R$, we have $d_P(u)=d_P(u-(k+1))=2$, where $P:=\{u,u+1,\dots ,u-(k+1)\}$.
		The other case is similar.
	\end{proof}

An important consequence of \cref{lemma:if not L then neighbor is R} is that if $G$ is a Hamiltonian cubic graph with no good shift, then we can define subsets of vertices that must be either in $L$ or in $R$.
To define such subsets we introduce the following notation.

\begin{definition}
		Let $G=(V,E)$ be a Hamiltonian cubic graph and $u \in V$.
        We define the vertex subset $\u:=\{v \in V : v \equiv u \ (\mathrm{mod}\ (k+1))\}$ where $k := \ceil{\frac{|V|-1}{3}}$.
\end{definition}

	\begin{corollary}\label[corollary]{corollary:propagation of L/R}
		Let $G=(V,E)$ be a Hamiltonian cubic graph with no good shift and $u\in V$:
        \begin{compactitem}
			\item if  $u \notin R$ then $\u \subseteq L$,
			\item if $u \notin L$, then $\u \subseteq R$,
			\item $|L|=|R|=\frac{|V|}{2}$.
		\end{compactitem}
	\end{corollary}

	\begin{proof}
		First, notice that for any integer $\delta\geq 1$, $u-\delta \cdot (k+1)\equiv u-\delta \cdot (k+1)+|V|\cdot \delta \cdot (k+1) \ (\mathrm{mod}\ |V|) \equiv u+\delta \cdot (|V|-1) \cdot (k+1) \ (\mathrm{mod}\ |V|)$.
        Moreover, $u\equiv u+ |V|\cdot (k+1) \ (\mathrm{mod}\ |V|)$.
        Thus, we have $\{u - \delta \cdot (k+1): \delta \geq 1, \delta \in \mathbb{N}\}=\{u + \delta \cdot (k+1): \delta \geq 1, \delta \in \mathbb{N}\}=\u$.

		Now, if $u \notin R$, then, with our assumption that $G$ has no good shift and \cref{lemma:if not L then neighbor is R}, we derive that $\u=\{u + \delta \cdot (k+1): \delta \geq 1, \delta \in \mathbb{N}\}\subseteq L$.
		Symmetrically, if $u \notin L$, then $\{u - \delta \cdot (k+1): \delta \geq 1, \delta \in \mathbb{N}\}\subseteq R$.

		This implies that for any vertex $u\in V$, either $u\in L$ or $u\in R$.
        Finally, since $u\in L$ if and only if $c(u)\in R$ and $u\in R$ if and only if $c(u)\in L$, then it is obvious that $|L| = |R| = \frac{|V|}{2}$.
	\end{proof}

	Let $G=(V,E)$ be a Hamiltonian cubic graph with no good shift and $d:=\gcd(k+1,|V|)$, where $\gcd(k+1,|V|)$ is the greatest common divisor of $(k+1)$ and $|V|$.
	We show that $V$ can be partitioned into $d$ subsets of vertices $<\!0\!>$, $<\!1\!>$, \dots, $<\!d-1\!>$.
	This partition will be useful to find an \emph{almost good shift} $P$ and a vertex to remove from $P$ in order to obtain a PDS in $G$.
	This result comes from a basic property of the cyclic group $\mathbb{Z}/n\mathbb{Z}$ that we recall in the following lemma.

\begin{lemma}\label[lemma]{lemma:gcd}
	Let $\alpha \geq 1$ and $\beta \geq 1$ be positive integers, and $d:=\gcd(\alpha,\beta)$.
	If all integers are considered $\mathrm{mod}\ \alpha$, then $\{0,1,\dots , \alpha-1\}=\cup_{i\in \{0,1,\dots , d-1\}}<\!i\!>$ where $<\!i\!>:=\{l : l \equiv i \ (\mathrm{mod}\ \beta)$ and $l \in \{0,1,\dots,\alpha-1\}\}$.
    Moreover, for any $i,j\in \{0,1,\dots , d-1\}$ with $i \neq j$, $<\!i\!>\cap <\!j\!>=\emptyset$.
\end{lemma}
\begin{proof}
	First, we prove that for any $u \geq d$, $u \in <\!i\!>$ for some $i\in \{0,1,\cdots , d-1\}$.
	Let $u \geq d$.
	Then there exist two integers $a,b$ with $b \leq d-1$, such that $u=a\cdot d+b$.
	Moreover, there exist two integers $c,f$ such that $c\cdot \beta+f \cdot \alpha=d$ since $d=\gcd(\alpha,\beta)$.
	Then, $u=a \cdot c \cdot \beta +a \cdot f \cdot \alpha +b \equiv b + a\cdot c \cdot \beta\ (\mathrm{mod}\ \alpha)$.
	Thus, $u\in <\!b\!>$ with $b\leq d-1$.
	This proves that any integer is in a set $<\!i\!>$ for some $i\leq d-1$, \textit{i.e.}\ $\{0,1,\cdots , \alpha-1\}=\cup_{i\in \{0,1,\cdots , d-1\}}<\!i\!>$.

	To prove the second part of the statement, we first show that $\alpha = |\!<\!u\!>\!| \cdot d$ for any $u\in \{0,1,\cdots ,d-1\}$.
	Let $u\in \{0,1,\cdots ,d-1\}$ and $p \geq 1$ be the smallest integer such that $u+p \cdot \beta \equiv u \ (\mathrm{mod}\ \alpha$).
	Notice that $|\!<\!u\!>\!|=p$ and let us show that $\alpha=p \cdot d$.
	Let $\alpha',\beta'$ be two integers such that $\alpha=\alpha' \cdot d$, $\beta=\beta' \cdot d$ and $\gcd(\alpha',\beta')=1$.
	We prove that $\alpha'=p$ by verifying that $\alpha'$ divides $p$ and $p$ divides $\alpha'$.
	First, notice that $u + \alpha' \cdot \beta=u + \alpha' \cdot k'\cdot d = u + \alpha \cdot \beta'\equiv u\ (\mathrm{mod}\ \alpha)$.
	Thus, $p$ divides $\alpha'$.
	On the other hand, recall that $u+p\cdot \beta \equiv u\ (\mathrm{mod}\ \alpha)$ and notice that $u+p\cdot \beta = u+p \cdot \beta' \cdot d$, then $p \cdot \beta' \cdot d \equiv 0\ (\mathrm{mod}\ \alpha)$.
	This implies that $\alpha$ divides $p\cdot \beta' \cdot d$, and thus $\alpha'$ divides $p \cdot \beta'$.
	Since $\gcd(\alpha',\beta')=1$, $\alpha'$ divides $p$.
	Now, notice that two sets $\!<\!i\!>\!$, $\!<\!j\!>\!$ for some integers $i,j$ are either equal or disjoint.
	Since for any $u \in \{0,1,\cdots,\alpha-1\}$ we have $|\!<\!u\!>\!| = \frac{\alpha}{d}$, then obviously all sets $\!<\!i\!>\!$, $i\in \{0,1,\cdots , d-1\}$ are disjoints.
\end{proof}

In the following lemma we summarize the possible values of $\gcd(n,k+1)$ for some specific values of $n$ and $k$.

	\begin{lemma}\label[lemma]{lemma:gcd-cases}
		Let $n$ be an even integer, $n\geq 4$. Then:
		\begin{compactitem}
			\item if $n=3k-1$, then $\gcd(n,k+1)\in \{2,4\}$,
			\item if $n=3k$, then $\gcd(n,k+1)\in \{1,3\}$,
			\item if $n=3k+1$, then $\gcd(n,k+1) =2$.
		\end{compactitem}
	\end{lemma}
  \begin{proof}
      Consider the case $n = 3k-1$, then
      $d := \gcd(k+1,3k-1) = \gcd(k+1,3k-1-2(k+1)) = \gcd(k+1, k-3) = \gcd(4,k-3)$.
      As $n$ is even, then $k$ is odd and $d \in \{2,4\}$.
      The other cases can be proved using the same reasoning.
  \end{proof}
	Firstly, we show that if $|V|=3k$, then there is always a good shift in $G$.

    \begin{corollary}\label[corollary]{corollary:gcd-odd}
		Let $G$ be a Hamiltonian cubic graph with $3k$ vertices, $k\geq 2$.
		Then $G$ has a good shift.
	\end{corollary}

	\begin{proof}
		Suppose by contradiction that there is no good shift in $G=(V,E)$.
		Notice that if $|V|=3k$, then $k = \ceil{\frac{|V|-1}{3}}$. Let $d:=\gcd(k+1,|V|)$.
		From \cref{lemma:gcd-cases} we get $d \in \{1,3\}$.
		According to \cref{corollary:propagation of L/R}, $|L|=|R|=\frac{|V|}{2}$.
		If $d=1$, then $V=\;<\!0\!>$ (\cref{lemma:gcd}), and hence $V=L$ or $V=R$, which is impossible.
		If $d=3$, then $|V|=\;<\!0\!>\cup <\!1\!> \cup <\!2\!>$ (\cref{lemma:gcd}).
		According to \cref{corollary:propagation of L/R}, $<\!i\!>\subseteq L$ or $<\!i\!>\subseteq R$ for any $i\in \{0,1,2\}$, and thus $|R|\neq |L|$, which is not possible.
	\end{proof}

	From \cref{lemma:gcd} and \cref{lemma:gcd-cases}, if a Hamiltonian cubic graph $G=(V,E)$ has no good shift, then $V$ can be written as $V=\;<\!0\!> \cup <\!1\!> \cup <\!2\!> \cup <\!3\!>$ (we may have $<\!0\!>\;=\;<\!2\!>$ and $<\!1\!>\;=\;<\!3\!>$). Hence, those graphs can be split into two categories:
	\begin{compactitem}
		\item \textit{type RLRL:} for any vertices $i,i+1$ with $i\in V$, we have $i\in L$ and $i+1\in R$, or $i\in R$ and $i+1\in L$.
		In this case, we always assume without loss of generality that $R=\;<\!0\!> \cup <\!2\!>$ and $L=\;<\!1\!> \cup <\!3\!>$.
		\item \textit{type RRLL:} there exist two vertices $i,i+1$ with $i\in V$ such that $i,i+1\in L$ or $i,i+1\in R$.
		In this case, we always assume without loss of generality that $R=\;<\!0\!> \cup <\!1\!>$ and $L=\;<\!2\!> \cup <\!3\!>$.
	\end{compactitem}
	Now, we show that if a Hamiltonian cubic graph $G$ has no good shift, then there exists
    an almost good shift $P$ in $G$ (\cref{lemma:goodshift+1}) and a vertex $v \in P$ such that $G[P\setminus\{v\}]$ is a PDS (\cref{lemma:hamiltonian<20} and \cref{theorem:hamiltoniancubic}).

	\begin{lemma}\label[lemma]{lemma:goodshift+1}
	Any Hamiltonian cubic graph with no good shift has an almost good shift.
	\end{lemma}
  \begin{proof}
      Let $G=(V,E)$ be a Hamiltonian cubic graph with no good shift, $k = \ceil{\frac{|V|-1}{3}}$ and $d := \gcd(k+1,|V|)$.
      Since $G$ has no good shift, according to \cref{lemma:gcd-cases} and \cref{corollary:gcd-odd}, $d \in \{2,4\}$ and $|V| = 3k-1$ or $|V| = 3k+1$.
      From \cref{corollary:propagation of L/R}, we know that each vertex in $V$ belongs to either $L$ or $R$.
      \begin{compactitem}
          \item Case 1: \textit{$G$ is of type RLRL}.
          Let $P:=\{0,1,\cdots,-k\}$.
          Since $|V|$ is even, then $|P|$ is even.
          Therefore, since two vertices $i,i+1 \in P$ do not both belong to $L$ or $R$, then the vertex $-k$ belongs to $L$.
          Then the set $P$ fulfills the requirements.

          \item Case 2: \textit{$G$ is of type RRLL}.
          Consider the set $P:=\{1,2,\cdots , -k+1\}$.
          According to \cref{lemma:gcd-cases}, since $d=4$, $|V|=3k-1$.
          Hence, $-k+1=2-(k+1)\in \;<\!2\!>$.
          Thus, $-k+1\in L$ and $P$ fulfills the requirements.
      \end{compactitem}
  \end{proof}

	Recall that the graphs $H_1$ and $H_2$ from  \cref{fig:counterexample-hamiltonian} have no proportionally dense subgraph of the maximum possible size.
	In \cref{theorem:hamiltoniancubic},
	we show that these are the only cubic Hamiltonian graphs with this property.

	Before proving the main theorem, we first deal with small graphs ($|V|<20$) that are particular cases that need to be treated independently.

	\begin{lemma}\label{lemma:hamiltonian<20}
		Let $G=(V,E)$ be a Hamiltonian cubic graph not isomorphic to $H_1$ or $H_2$ with $|V|<20$. Then there exists a PDS of size $\floor{\frac{2\cdot|V|+1}{3}}$ in $G$.
	\end{lemma}
    \begin{proof}
        Let $k = \ceil{\frac{|V|-1}{3}}$.
        Since $G$ is cubic, its number of vertices is even.
        From \cref{lemma:gcd-cases}, $\gcd(k+1,|V|)\in \{1,2,3,4\}$. If $\gcd(k+1,|V|)\in \{1,3\}$, then there exists a good shift from \cref{corollary:gcd-odd}.
        We then suppose that $\gcd(k+1,|V|)\in \{2,4\}$.
        The following cases remain:

        \begin{compactitem}
            \item If $|V|=4$, then $G$ is the complete graph $K_4$, and any set of $3$ vertices induces a PDS of size $\lfloor \frac{2\cdot 4+1}{3} \rfloor$.
            
            \item If $|V|=8$, we claim that $G$ must have a good shift. By contradiction, suppose that $G$ has no good shift.
            If $G$ is of \textit{type RRLL} then $G$ is isomorphic to $H_1$, and if $G$ is of \textit{type RLRL} then $G$ is isomorphic to $H_2$, which is impossible since we assumed that $G$ is not isomorphic to $H_1$ or $H_2$.

            \item If $|V|=10$ and $G$ has no good shift, since $\gcd(k+1,|V|)=2$, $G$ is necessarily of \textit{type RLRL} and $c(0)=3$, $c(1)=8$, $c(2)=5$, $c(4)=7$, $c(6)=9$.
            In this case, $V\setminus \{0,6,9\}$ induces a PDS of size $\lfloor \frac{2\cdot 10+1}{3} \rfloor$.

            \item If $|V|=14$, if $G$ has no good shift, since $\gcd(k+1,|V|)=2$, then $G$ is necessarily of \textit{type RLRL}.
            Following \cref{lemma:goodshift+1}, let $P:=\{0,1,\cdots ,9\}$ be an almost good shift and:

            \begin{compactitem}
                \item If $c(6)\neq 9$, notice that $c(7),c(5)\in P$ (since $5,7\in L$) and $c(6)\in V\setminus P$. Thus, $G[P\setminus \{6\}]$ is a PDS of size $\lfloor \frac{2\cdot 14+1}{3} \rfloor$.
                If $c(3)\neq 0$, the case is symmetrical.
                
                \item If $c(3)= 0$ and $c(6)= 9$, notice that $c(3)\in P$, $c(5)\in P$ and $d_P(c(4))=3$ since $c(4)\neq 9)$. Thus, $G[P\setminus \{4\}]$ is a PDS of size $\lfloor \frac{2\cdot 14+1}{3} \rfloor$.
            \end{compactitem}
            \item If $|V|=16$, if $G$ has no good shift, since $\gcd(k+1,|V|)=2$, $G$ is necessarily of \textit{type RLRL}. Following \cref{lemma:goodshift+1}, let $P:=(0,1,\cdots ,-k)$ be an almost good shift. Since $0\in R$, we have either $c(0)=3$ or $c(0)=5$. In each case, the graph is completely determined due to the constraints. In the first case, $P\setminus \{4\}$ induces a PDS of size $\lfloor \frac{2\cdot 16+1}{3} \rfloor$. In the second case, $P\setminus \{3\}$ induces a PDS of size $\lfloor \frac{2\cdot 16+1}{3} \rfloor$.
        \end{compactitem}

        In each case, if $G$ is not isomorphic to $H_1$ or $H_2$, then either $G$ has a good shift which is a PDS of size $\floor{\frac{2\cdot|V|+1}{3}}$, or we give a PDS of such size.
    \end{proof}

	\begin{theorem}\label[theorem]{theorem:hamiltoniancubic}
		Let $G=(V,E)$ be a Hamiltonian cubic graph not isomorphic to $H_1$ or $H_2$.
		Then there exists a connected PDS of size $\floor{\frac{2\cdot|V|+1}{3}}$ in $G$.
	\end{theorem}
	\begin{proof}
        If $|V|<20$, then there is a PDS of size $\floor{\frac{2\cdot|V|+1}{3}}$ in $G$ from \cref{lemma:hamiltonian<20}. 		  Now we suppose that $|V|\geq 20$, which implies that $k:=\lceil \frac{|V|-1}{3}\rceil\geq 7$.

        From \cref{lemma:gcd-cases}, $\gcd(k+1,|V|)\in \{1,2,3,4\}$.
        If $\gcd(k+1,|V|)\in \{1,3\}$, then there exists a good shift (\cref{corollary:gcd-odd}).

        We suppose that $\gcd(k+1,|V|)\in \{2,4\}$.
        If $G$ contains a good shift, then the proof is done. Notice that in such case, the PDS is obviously connected.
        Now, we assume that $G$ has no good shift.
        We prove that given an almost good shift $P$, there exists a vertex $u^* \in P$ such that $G[P\setminus \{u^*\}]$ is a PDS.
        Observe that such vertex $u^*$ exists if and only if $c(u^*-1),c(u^*+1)\in P$, and either $c(u^*) \in V \setminus P$ or $d_P(c(u^*))=3$.

        \begin{compactitem}
            \item If $G$ is of \textit{type RLRL}, then $R=\;<0>\cup <2>$ and $L=\;<1>\cup <3>$.
            According to \cref{lemma:goodshift+1}, the set $P := \{0,1,2,\cdots ,-k\}$ is an almost good shift and $0\in R, 1\in L$.
            Since $2\in R$ and $4\in R$, then $c(2)\in P$ and $c(4)\in P$.
            If $c(3)\neq 0$, then $c(3) \in V \setminus P$ since $3\in L$.
            Thus, $G[P\setminus \{3\}]$ is a PDS of size $\lfloor \frac{2\cdot|V|+1}{3} \rfloor$.
            Symmetrically, if $c(-k-3)\neq -k$, then $c(-k-3) \in V \setminus P$ since $3\in R$. Thus, $G[P\setminus \{-k-3\}]$ is a PDS of size $\lfloor \frac{2\cdot|V|+1}{3} \rfloor$.
            On the other hand, if $c(3)=0$ and $c(-k-3)= -k$, then $c(k-1)\neq -k$ and $c(k-1)\in P$.
            Moreover, since $k-3\in R$ then $c(k-3)\in P$.
            Therefore, $c(k-2) \in V \setminus P$ or $d_P(c(k-2))=3$ (since $k\geq 7$, $k-2\neq 3$ and $c(k-2)\neq 0)$.
            Thus, $G[P\setminus \{k-2\}]$ is a PDS of size $\lfloor \frac{2\cdot|V|+1}{3} \rfloor$.
            Notice that the resulting PDS is connected.
            Indeed, let $v$ be the vertex we removed from the path $\{0,1,\cdots ,-k\}$.
            It is easy to see that, either $c(v-1)\in \{v+1,v+2,\cdots , -k\}$, or $c(v+1)\in \{0,1,\cdots , v-1\}$ since the graph is of type $RLRL$, and thus the PDS is connected.

            \item If $G$ is of \textit{type RRLL}, then $R=\;<\!0\!>\cup <\!1\!>$ and $L=\;<\!2\!>\cup<\!3\!>$.
            According to \cref{lemma:goodshift+1}, the set $P := \{1,2,\cdots ,-k+1\}$ is
            an almost good shift
            and $1\in R, 2\in L, -k\in R, -k+1 \in L$.
            Since $k+1\in \;<0>$ and $k+2\in \;<1>$, we necessarily have $k-1,k\in L$ and $k+1,k+2\in R$.
            In this case, notice that since $k\geq 7$, $\{k-3,k-2,k-1\}\in P$.
            Moreover, $k-3,k-2\in R$, which implies $c(k-3),c(k-2)\in P$.
            We show that either $c(k-1)\in P$ or $c(k)\in P$.
            Suppose that $c(k)\notin P$.
            Then since $k\in L$, we have $c(k)=0$.
            Since $k-1\in L$, we have $c(k-1)\in \{-1,0,1,\cdots ,k-3\}$.
            Since $0=c(k)$ and $-1\in L$, then $c(k-1)\neq -1$ and $c(k-1)\neq 0$.
            Thus, $c(k-1)\in \{1,2,...,k-3\}\subset P$.
            Thus, either $c(k-1)\in P$ or $c(k)\in P$.
            Now, if $c(k-1)\in P$, then since $c(k-3)\in P$, the set $G[P\setminus \{k-2\}]$ is a PDS of size $\lfloor \frac{2\cdot|V|+1}{3} \rfloor$.
            Else, $c(k)\in P$ and then since $c(k-2)\in P$, the set $G[P\setminus \{k-1\}]$ is a PDS of size $\lfloor \frac{2\cdot|V|+1}{3} \rfloor$.
            Notice that the resulting PDS is connected.
            Indeed, let $v$ be the vertex we removed from the almost good path $\{1,2,\cdots ,-k+1\}$.
            Again, it is easy to verify that either $v=k-2$, and then $c(k-3) \in \{k-1,k,\cdots , -k+1\}$, or $v=k-1$, and then $c(k)\in \{1,2,\cdots ,k-2\}$ since the graph is of type $RRLL$.
            Thus the PDS is connected.
        \end{compactitem}
    \end{proof}

	According to \cref{lemma:upperbound-PDS}, a PDS in a cubic graph of order $n$ contains at most $\floor{ \frac{2n+1}{3} }$ vertices.
	Thus, we obtain the following corollary.

	\begin{corollary}
	Let $G$ be a Hamiltonian cubic graph with a given Hamiltonian cycle.
	Then a connected proportional dense subgraph of maximum size in $G$ can be found in linear time.
	\end{corollary}

\section{Conclusion and open problems}\label{section:conclusion}

We prove that {\MClong} is \APX-hard even on split graphs, and \NP-hard on bipartite graphs, whether the PDS is required to be connected or not.
Furthermore, the problem is proved to be $(2-\frac{2}{\Delta+1})$-approximable, where $\Delta$ is the maximum degree of the graph.
We also show that deciding if a PDS is inclusion-wise maximal is \textsf{co-NP}-complete, even on bipartite graphs.
Nevertheless, {\MC} can be solved in linear time on Hamiltonian cubic graphs if a Hamiltonian cycle is given.

However, the complexity of finding a PDS of maximum size in cubic graphs remains unknown.
More specifically, the question whether a PDS of size $\floor{ \frac{2n+1}{3} }$ always exists in a cubic graph is still open (except for the two graphs given in \cref{fig:counterexample-hamiltonian}).
Also, \cref{algo:PDS at least n/2} returns a PDS of size $\ceil{\frac{n}{2}}$ or $\ceil{\frac{n}{2}}+1$ (in linear time), but the PDS may not be connected.
An interesting open question is whether there is always a connected PDS of size at least $\ceil{\frac{n}{2}}$.
Finally, the parameterized complexity of finding a PDS of size at least $\ceil{\frac{n}{2}} + k$ is unknown.
 
\bibliographystyle{abbrvnat}
\bibliography{ref}

\begin{thebibliography}{14}
\providecommand{\natexlab}[1]{#1}
\providecommand{\url}[1]{\texttt{#1}}
\expandafter\ifx\csname urlstyle\endcsname\relax
  \providecommand{\doi}[1]{doi: #1}\else
  \providecommand{\doi}{doi: \begingroup \urlstyle{rm}\Url}\fi

\bibitem[Andersen and Chellapilla(2009)]{andersen2009finding}
R.~Andersen and K.~Chellapilla.
\newblock Finding dense subgraphs with size bounds.
\newblock In \emph{6th International Workshop on Algorithms and Models for the
  Web-Graph (WAW)}, volume 5427 of \emph{LNCS}, pages 25--37. Springer, 2009.
\newblock \doi{10.1007/978-3-540-95995-3_3}.

\bibitem[Anderson(2007)]{Anderson2007}
R.~Anderson.
\newblock Finding large and small dense subgraphs.
\newblock arXiv:cs/0702032, 2007.

\bibitem[Asahiro et~al.(2002)Asahiro, Hassin, and Iwama]{asahiro2002complexity}
Y.~Asahiro, R.~Hassin, and K.~Iwama.
\newblock Complexity of finding dense subgraphs.
\newblock \emph{Discrete Applied Mathematics}, 121\penalty0 (1-3):\penalty0
  15--26, 2002.
\newblock \doi{10.1016/S0166-218X(01)00243-8}.

\bibitem[Bazgan et~al.(2018{\natexlab{a}})Bazgan, Chleb\'ikov\'a, and
  Dallard]{BazChleDal2018}
C.~Bazgan, J.~Chleb\'ikov\'a, and C.~Dallard.
\newblock Graphs without a partition into two proportionally dense subgraphs.
\newblock arXiv:cs/1806.10963, 2018{\natexlab{a}}.

\bibitem[Bazgan et~al.(2018{\natexlab{b}})Bazgan, Chleb\'ikov\'a, and
  Pontoizeau]{BCP17}
C.~Bazgan, J.~Chleb\'ikov\'a, and T.~Pontoizeau.
\newblock Structural and algorithmic properties of $2$-community structures.
\newblock \emph{Algorithmica}, 80\penalty0 (6):\penalty0 1890--1908,
  2018{\natexlab{b}}.
\newblock \doi{10.1007/s00453-017-0283-7}.

\bibitem[Chleb\'\i{}k and Chleb\'\i{}kov\'a(2006)]{chleb2006apx}
M.~Chleb\'\i{}k and J.~Chleb\'\i{}kov\'a.
\newblock Complexity of approximating bounded variants of optimization
  problems.
\newblock \emph{Theoretical Computer Science}, 354:\penalty0 320--338, 2006.
\newblock \doi{10.1016/j.tcs.2005.11.029}.

\bibitem[Feige et~al.(2001)Feige, Peleg, and Kortsarz]{feige2001dense}
U.~Feige, D.~Peleg, and G.~Kortsarz.
\newblock The dense $k$-subgraph problem.
\newblock \emph{Algorithmica}, 29\penalty0 (3):\penalty0 410--421, 2001.
\newblock \doi{10.1007/s004530010050}.

\bibitem[Goldberg(1984)]{Goldberg:CSD-84-171}
A.~V. Goldberg.
\newblock Finding a maximum density subgraph.
\newblock Technical Report UCB/CSD-84-171, EECS Department, University of
  California, Berkeley, 1984.

\bibitem[Karp(1972)]{Karp1972}
R.~M. Karp.
\newblock Reducibility among combinatorial problems.
\newblock In \emph{Complexity of computer computations}, pages 85--103.
  Springer, 1972.
\newblock \doi{10.1007/978-1-4684-2001-2_9}.

\bibitem[Khuller and Saha(2009)]{KhullerSaha2009}
S.~Khuller and B.~Saha.
\newblock On finding dense subgraphs.
\newblock In \emph{36th International Colloquium on Automata, Languages and
  Programming (ICALP)}, volume 5555 of \emph{LNCS}, pages 597--608.
  Springer-Verlag, 2009.
\newblock \doi{10.1007/978-3-642-02927-1_50}.

\bibitem[Kristiansen et~al.(2004)Kristiansen, Hedetniemi, and
  Hedetniemi]{kristiansen04alliances}
P.~Kristiansen, S.~M. Hedetniemi, and S.~T. Hedetniemi.
\newblock Alliances in graphs.
\newblock \emph{Journal of Combinatorial Mathematics and Combinatorial
  Computing}, 48:\penalty0 157--177, 2004.

\bibitem[Olsen(2013)]{Ols13}
M.~Olsen.
\newblock A general view on computing communities.
\newblock \emph{Mathematical Social Sciences}, 66\penalty0 (3):\penalty0
  331--336, 2013.
\newblock \doi{10.1016/j.mathsocsci.2013.07.002}.

\bibitem[Robinson and Wormald(1992)]{robinson1992almost}
R.~W. Robinson and N.~C. Wormald.
\newblock Almost all cubic graphs are hamiltonian.
\newblock \emph{Random Structures \& Algorithms}, 3\penalty0 (2):\penalty0
  117--125, 1992.
\newblock \doi{10.1002/rsa.3240030202}.

\bibitem[Rodr{\'\i}guez-Vel{\'a}zquez et~al.(2009)Rodr{\'\i}guez-Vel{\'a}zquez,
  Yero, and Sigarreta]{rodriguez2009defensive}
J.~A. Rodr{\'\i}guez-Vel{\'a}zquez, I.~G. Yero, and J.~M. Sigarreta.
\newblock Defensive $k$-alliances in graphs.
\newblock \emph{Applied Mathematics Letters}, 22\penalty0 (1):\penalty0
  96--100, 2009.
\newblock \doi{10.1016/j.aml.2008.02.012}.

\end{thebibliography}
\end{document}